\documentclass{amsart}

\usepackage[T1]{fontenc}
\usepackage[utf8]{inputenc}
\usepackage{graphicx}
\usepackage[dvipsnames]{xcolor}
\usepackage{tgtermes}
\usepackage{centernot}
\usepackage{tikz}
\usetikzlibrary{positioning}
\usepackage[
pdftitle={non-integer moments odd orthogonal ensemble}, 
pdfauthor={E. Alvarez},
colorlinks=true,linkcolor=blue,urlcolor=blue,citecolor=blue,bookmarks=true,
bookmarksopenlevel=2]{hyperref}
\usepackage{amsmath,amssymb,amsthm,textcomp, tabu}
\usepackage{enumerate}
\usepackage{multicol}
\newtheorem{innercustomthm}{Theorem}
\newenvironment{customthm}[1]
  {\renewcommand\theinnercustomthm{#1}\innercustomthm}
  {\endinnercustomthm}

\usepackage{geometry}
\newcommand{\abs}[1]{\left\lvert #1 \right\rvert}
\newtheoremstyle{mytheor}
    {1ex}{1ex}{\normalfont}{0pt}{\scshape}{.}{1ex}
    {{\thmname{#1 }}{\thmnumber{#2}}{\thmnote{ (#3)}}}

\newcommand{\N}{\mathbb{N}}
\newcommand{\R}{\mathbb{R}}

\newtheorem{defi}{Definition}
\newtheorem{theorem}{Theorem}
\newtheorem{lemma}{Lemma}
\newtheorem{prop}{Proposition}



\newcommand{\dx}[1]{\mathrm{d{#1}}}
\newcommand{\dX}{\mathrm{d_{Haar}(X)}}
\usepackage{fancyhdr,lastpage}
\let\bs=\bigskip
\begin{document}
\title{Asymptotics of non-integer moments of the logarithmic derivative of characteristic polynomials over $SO(2N+1)$}

\author{E. Alvarez}  
\email[University of Bristol]{emilia.alvarez@bristol.ac.uk}
\author{P. Bousseyroux}  
\email[\'Ecole
Normale Sup\'erieure de Paris]{pierre.bousseyroux@ens.psl.eu}
\author{N.C. Snaith }
\email[University of Bristol]{n.c.snaith@bris.ac.uk}

\begin{abstract}
    This work computes the asymptotics of the non-integer moments of the logarithmic derivative of characteristic polynomials of matrices from the $SO(2N+1)$ ensemble. It follows from work of Alvarez and Snaith who computed the asymptotics of the integer moments of the same statistic over both $SO(N)$ ensembles as well as the $USp(2N)$ ensemble. 
 \end{abstract}

\maketitle
The following work develops further the argument made in \cite{AS}, that the asymptotics of the integer moments of the logarithmic derivative of characteristic polynomials from the classical compact ensembles evaluated at a point approaching 1, are strongly affected by the term corresponding to the eigenvalue closest to 1. Since the logarithmic derivative has the characteristic polynomial in the denominator, an eigenvalue at 1 causes a singularity in the logarithmic derivative when evaluated at $s=1$. This behaviour is observed most clearly in the odd orthogonal ensemble $SO(2N+1)$, since all matrices in this ensemble have a fixed eigenvalue at the point 1. In this paper we show in the case of $SO(2N+1)$ that the asymptotics of the moments of the logarithmic derivative are dominated by the eigenvalue at 1 for non-integer moment parameter as well. 

The motivation to study non-integer moments comes from number theory: an approach to the Riemann hypothesis is through a formula of Littlewood that requires knowledge of arbitrary moments of the derivative of the Riemann zeta-function, where the moment parameter is continuous in a subset of $\mathbb{R}$. The connection between number theory and random matrices is well established, see for example \cite{KS,kn:hughes03,CFKRS} or the review papers \cite{Conrey,kn:snaith10,kn:keasna03}, therefore, as random matrix theorists we study the analogous question which is \textit{can we compute the non-integer moments of the derivative of characteristic polynomials of random matrices?} More often than not, we are able to prove theorems on the random matrix side which shed light on conjectures for the number theoretic analogues. In this case, however, an exact form on the random matrix side has remained elusive. In his thesis, Hughes \cite{kn:hug01} conjectured the large $N$ limit for the joint moments of characteristic polynomials and their derivatives over $U(N)$ for general moment parameters. The best progress on non-integer moments of the derivative before 2020 was the work of Winn in \cite{Winn}, where he obtained asymptotics of the odd moments over the unitary ensemble while evaluating the $2k^{th}$ moment, so effectively Winn was able to get a formula for $k$ a half integer. In \cite{AKW}, Assiotis et al. use the form obtained by Winn, and further show that Hughes' conjectured limit does indeed exist. They show it is equal to the moments of a random variable related to a point process and to Hua-Pickrell measures, for non-integer moment parameters. This is however difficult to expand in an exact form. The ensembles $U(N)$, $SO(N)$ and $USp(2N)$ with Haar measure are examples of $\beta=2$ ensembles (where the eigenvalues display quadratic repulsion from their neighbours), and in \cite{for22} Forrester is able to generalise Winn's approach to consider non-integer moments of general $\beta > 0$. 

In this work, we are not attacking the non-integer moments of the derivative of characteristic polynomials themselves, but rather we are looking at the moments of the logarithmic derivative. These are interesting in their own right, they have analogues in number theory with moments of logarithmic derivatives of $L$-functions and they are related to the joint moments of characteristic polynomials and their derivatives. Indeed, the moments of the logarithmic derivative of the Riemann zeta function was first studied by Selberg to study primes in short intervals \cite{Selberg}. More recently, Goldston, Gonek and Montgomery showed that the moments of the logarithmic derivative of the Riemann zeta function are related to the pair correlation conjecture for the zeroes of the Riemann zeta function and to counting prime powers \cite{GGM}. Farmer et. al build on these methods in \cite{FGLL} to study almost primes by assuming the random matrix conjectures that the correlation functions of zeroes of $\zeta(s)$ agree, in appropriate scaling limits,  with the correlation functions of eigenvalues of Gaussian Unitary matrices. Finally, the distribution of the logarithmic derivative of $\zeta(s)$ has also been studied just off the critical line by Guo \cite{Guo} and Lester \cite{Lester}. These results are related to results from $U(N)$; see the works \cite{MM} and \cite{Basor} for the integer moments of the logarithmic derivative of characteristic polynomials over $U(N)$. 

Since number theorists are also interested in analogous results for families of $L$-functions which have heuristically been shown to behave like random matrices from the other classical compact ensembles (see \cite{KatzS}), it is of interest for us to consider the moments of the logarithmic derivative of characteristic polynomials from $SO(2N+1)$ as well. See the work of Ihara \cite{Ihara}, Matsumoto \cite{Mats} and Mourtada and Murty \cite{MouMur} on moments of the logarithmic derivatives of Dirichlet $L$-functions. Other probabilistic models, different from random matrix theory, have also been developed in \cite{gransound, hattori} to study $L$-functions and their logarithmic derivatives. In \cite{hamieh1} Hamieh and McClenagan study the limiting distribution of the logarithmic derivative of Dirichlet $L$-functions varying the characters of the $L$-function. They also show that the probabilistic model developed in \cite{gransound} for studying $L$-functions in the same regime is a good model for the logarithmic derivative. In \cite{lamzouri1} Lamzouri also studies the logarithmic derivative of $L$-functions evaluated at the point 1 in a probabilistic regime. 

We focus here on the $SO(2N+1)$ ensemble as this is the only ensemble with a fixed eigenvalue. Our methods could be amenable to adaptation for the $U(N)$, $SO(2N)$ and $USp(2N)$ ensembles if one were able to isolate the contribution of the first eigenvalue for each matrix in the ensemble. 

\section{Preliminaries}
We define here our main objects of consideration and recall some key tools we will use in the proof below. 

\begin{defi}
    The set of matrices denoted by $SO(2N+1)$ is defined as
\begin{equation}
    SO(2N+1) := \{X \in GL(2N+1, \R) : XX^{T} = X^{T}X = I_N, \quad \det(X) = 1 \}. 
\end{equation} 
\end{defi}
A matrix $X \in SO(2N+1)$ has $\det(X) = 1$, $N$ pairs of complex conjugate eigenvalues and an additional eigenvalue exactly at 1. 

\begin{defi}
We define the characteristic polynomial of a matrix $X$ to be 
\begin{equation}
    \Lambda_X(s) = \det(I - sX^{*}) ,
\end{equation}
where $X^*$ is the conjugate transpose of $X$. 
\end{defi}
The eigenvalues of a matrix $X$ from $SO(2N+1)$ come in complex conjugate pairs $e^{i\theta_1},e^{-i\theta_1} \dots, e^{i\theta_N}, e^{-i\theta_N}$, plus one eigenvalue at 1, and so its characteristic polynomial can also be expressed as:
\begin{equation}
   \Lambda_X(s) = (1-s)\prod_{j=1}^N (1 - se^{-i\theta_j})(1 - se^{i\theta_j}).
\end{equation}
 We will be computing moments by averaging over the set $SO(2N+1)$ with respect to Haar measure, denoted $\dX$, which is the unique translation invariant measure. We refer to this as the odd orthogonal ensemble.  
 
 We recall the infinity (or max) norm over an interval $\mathcal{I}$ denoted
\begin{equation}
    ||f(x)||_{\infty, \mathcal{I}}
\end{equation} is given by the supremum of $f(x)$ for $x \in \mathcal{I}$.
We also recall that a homographic function is of the form 
\begin{equation}\label{homographic}
    h(x) = \frac{ax + b}{cx + d}. 
\end{equation} It has a singularity at $\tfrac{-d}{c}$ and if $ad - bc < 0$, it is decreasing on the intervals $(-\infty, \tfrac{-d}{c})$ and $(\tfrac{-d}{c},\infty)$. 

Finally, we use the following generalization of binomial coefficients:

\begin{equation}
    \binom{\alpha}{n}= \frac{1}{n!}\prod_{j = 0}^{n-1} (\alpha - j),
\end{equation}
where $\alpha\in \R$ and $n\in \N$.

Writing out the logarithmic derivative in terms of its eigenvalues ($1,e^{\pm i\theta_1}, \dots, e^{\pm i\theta_N}$), we can isolate the term that dominates as $s$ approaches 1. We have that 
\begin{align}\label{SO2N+1logder}
    \frac{\Lambda_X'}{\Lambda_X}(s)=\frac{-1}{1-s}+\sum_{n=1}^N \left( \frac{-e^{i\theta_n}}{1-se^{i\theta_n}}-\frac{e^{-i\theta_n}}{1-se^{-i\theta_n}} \right) \nonumber \\
    = \frac{-1}{1-s}+ \sum_{n=1}^N \frac{2s - 2\cos(\theta_n)}{s^2 - 2s\cos(\theta_n) + 1} \nonumber \\
    = \frac{-1}{1-s} \left( 1 + (s-1) \sum_{n=1}^N \frac{2s - 2\cos(\theta_n)}{s^2 - 2s\cos(\theta_n) + 1} \right).
\end{align} In \cite{AS}, we observe that the leading order behaviour, $\left(\frac{-N}{a}\right)^K$ in Theorem \ref{thm4}, comes entirely from substituting $s=e^{-a/N}$ into  $\frac{-1}{1-s}$ as $N\rightarrow \infty$ and raising to the $K^{th}$ power. Indeed, the asymptotic computation of integer moments is described in the following theorem.
\begin{theorem}[Alvarez and Snaith, 2021]\label{thm4}
  Let $\Lambda_X(s)$ denote the characteristic polynomial of a matrix $X \in SO(2N+1)$, the group of odd dimensional random orthogonal matrices with determinant 1 equipped with the Haar measure $\dX$. Let $K \in \mathbb{N}$, $\alpha = a/N$ where $a = o(1)$ as $N \to \infty$ and $\mathfrak{Re}(a)>0$. Then, as $N$ tends to $\infty$, the moments of the logarithmic derivative of $\Lambda_X(s)$ evaluated at $e^{-\alpha}$ are given by:  
  \begin{align}
    \int_{SO(2N+1)} \left(\frac{ \Lambda_X^{'}}{\Lambda_X }(e^{-\alpha})\right)^K \dX \nonumber \\
     = (-1)^K \left[ \left( \frac{N}{a} \right)^K - \frac{N^K}{a^{K-1}}K \right] + \mathcal{O}\left(\frac{N^{K-1}}{a^{K-1}}\right)  + \mathcal{O}\left( \frac{N^K}{a^{K-2}}  \right). 
\end{align}
\end{theorem} 

In the present work, we extend this result to non-integer moments. Our main result is Theorem \ref{absthm}. 

\begin{theorem}\label{absthm}
Let $\Lambda_X(s)$ denote the characteristic polynomial of a matrix $X \in SO(2N+1)$, the group of odd orthogonal random matrices of odd dimension and determinant 1, equipped with the Haar measure $\dX$. Let $\alpha = a/N>0$ with $a \in \mathbb{R}_{+}$, $a = o(1)$ as $N \to \infty$, and the moment parameter $K \in \mathbb{R}_{+}$. Then, 
\begin{equation}
\left(\frac{-a}{N}\right)^K\int_{SO(2N+1)} \left(\frac{\Lambda_X'}{\Lambda_X}(e^{-\alpha})\right)^K \dX \underset{N\to +\infty}{=} 1 - Ka + \mathcal{O}(a/N) + \mathcal{O}(a^2).
\end{equation}
\end{theorem}

This theorem proves that the main contribution of the moments of the logarithmic derivative of characteristic polynomials over $SO(2N+1)$ comes from the term corresponding to the eigenvalue at 1, even for non-integer moments. Consider the intuition behind the second term “$-Ka$” in the Taylor’s expansion; the following argument is not be precise but serves to give a feeling and moral justification for the result. 

We recall that we have set $s = e^{-a/N}$. Factoring out the leading order term $\frac{-1}{1-s}\sim \frac{-N}{a}$ in \eqref{SO2N+1logder}, we want to study 

\begin{equation}
    \int_{SO(2N+1)} \left(1 + (s-1)\sum_{n=1}^N \frac{2s - 2\cos(\theta_n)}{s^2 - 2s\cos(\theta_n) + 1}\right)^K \dX.
\end{equation}

If $N$ is fixed and none of the $\theta_j$'s are equal to $0$, we have that

\begin{equation}\label{limit}
    \sum_{n=1}^N \frac{2s - 2\cos(\theta_n)}{s^2 - 2s\cos(\theta_n) + 1}\underset{s\to 1}{\longrightarrow}\sum_{n=1}^N \frac{2 - 2\cos(\theta_n)}{2 - 2\cos(\theta_n)} = N.
\end{equation}

Now, from an statistical point of view, it is unlikely that many matrices will have an eigenvalue $e^{i\theta_n}$ near 1 because of the repulsion from the eigenvalue at 1 in the eigenvalue joint probability density function for $SO(2N+1)$ on $[0, \pi]^N$:
\begin{equation}
    \frac{2^{N^2}}{\pi^N N!} \Delta^2(\cos(\theta))\prod_{m=1}^N \sin^2(\theta_m/2).
\end{equation}
The Vandermonde $\Delta^2(\cos(\theta)):=\prod_{1\leq j < k\leq N} (\cos \theta_k - \cos\theta_j)^2$ causes a repulsion between the eigenvalues, and the factor $\sin^2(\theta_m/2)$ encodes the repulsion from the point 1.

By using \eqref{limit} and the repulsion, we are then tempted to write

\begin{equation}\label{intuition}
    \left(1 + (s-1)\sum_{n=1}^N \frac{2s - 2\cos(\theta_n)}{s^2 - 2s\cos(\theta_n) + 1}\right)^K \underset{N\to \infty}{\approx} (1 + (s-1)N)^K \underset{N\to \infty}{\approx} (1 - a)^K \approx 1 - Ka
\end{equation}
which gives us our result. 

Nevertheless, we need to be careful about $N$ because it is not fixed. We can be more precise with the following reasoning. We start by rewriting

\begin{equation}\label{summand2}
\sum_{n=1}^N \frac{2s - 2\cos(\theta_n)}{s^2 - 2s\cos(\theta_n) + 1}=\frac{N}{s} + \frac{s^2-1}{s}\sum_{n=1}^N \frac{1}{(1-s)^2 + 4s\sin^2(\theta_n/2)}.
\end{equation}

It can be shown that

\begin{equation}
   \int_{SO(2N+1)}  \sum_{n=1}^N \frac{1}{(1-s)^2 + 4s\sin^2(\theta_n/2)}  \dX= \int_{0}^{\pi} \frac{R_1(\theta)}{(1-s)^2 + 4s\sin^2(\theta/2)}d\theta
\end{equation}
where $R_1(\theta)$ is the $1$-point correlation function of eigenvalues of matrices from $SO(2N+1)$, because $R_1(\theta)$ contains information on the relative likelihood of finding an eigenangle at various values of $\theta$ as we sample matrices from the odd orthogonal ensemble. 

The behaviour of $R_1$ close to $0$ measures the strength of the repulsion from the point $1$. If $R_1(\theta)$ goes to $0$ fast enough when $\theta$ goes to $0$, then eigenangles near $0$ are unlikely and the above integral can be bounded by something which does not depend on $s$ or $N$ because we avoid the singularity $\frac{1}{4\sin^2(\theta/2)}$ when $s\to 1$. In Section \ref{last_section} we will see  that it is indeed the case. And so, by using \eqref{summand2}, the approximation

\begin{equation}
\sum_{n=1}^N \frac{2s - 2\cos(\theta_n)}{s^2 - 2s\cos(\theta_n) + 1} \approx \frac{N}{s} \approx{N}
\end{equation}
as $s\rightarrow 1$ seems to be more justified than \eqref{limit} even when $N$ is not fixed.

To make the above reasoning rigorous, in Section \ref{absthm} we will approximate $(1+x)^K$ by a polynomial: 

\begin{equation}\label{polynomialapproximation}
    (1+x)^K \approx \sum_{n=0}^{\lfloor K \rfloor} \binom{K}{n} x^n,
\end{equation}
where $x$ corresponds to the sum over the complex conjugate pairs of eigenvalues divided by the contribution from the eigenvalue at 1, i.e. $(s-1)\sum_{n=1}^N \frac{2s - 2\cos(\theta_n)}{s^2 - 2s\cos(\theta_n) + 1}$. Here $\lfloor K \rfloor$ is the floor of $K$, the largest integer  less than or equal to $K$. 

Then, we will expand this polynomial approximation in the basis $((1+x)^n)_{n\in \N}$ in order to use Theorem \ref{thm4} for integer moments.

This method of polynomial approximation can provide a more precise Taylor expansion if you assume that $a$ goes to $0$ fast enough. Indeed, by using some changes of variables and some good approximations of the $1-$ and $2$-point correlation functions, we improve Theorem \ref{absthm} in Section \ref{last_section} without applying Theorem \ref{thm4}.

\begin{theorem}\label{theorem2}
Let $\Lambda_X(s)$ denote the characteristic polynomial of a matrix $X \in SO(2N+1)$, the group of odd orthogonal random matrices of odd dimension and determinant 1, equipped with the Haar measure $\dX$. Let $m$ be an odd integer such as $m\geq 3$ and $m\geq \lfloor K\rfloor$. Let $\alpha = a/N>0$ with $a = o\left(\frac{1}{N^{m+2}}\right)$ as $N \to \infty$, and the moment parameter $K \in \mathbb{R}_{+}$. 

Then, 

\begin{multline}
\left(\frac{-a}{N}\right)^K\int_{SO(2N+1)} \left(\frac{\Lambda_X'}{\Lambda_X}(e^{-\alpha})\right)^K \dX \underset{N\to +\infty}{=} \\1 - Ka + \frac{Ka}{2N} + \frac{K(K+1)}{2}a^2 - \frac{K(K+1)a^2}{2N} + \frac{a^2}{N^2}K\frac{3K-1}{24} + \mathcal{O}(N^m a^3).
\end{multline}
\end{theorem}

\section{Proof of Theorem \ref{absthm}}
The proof of Theorem \ref{absthm} for general moments begins by first expanding the moment using Weyl's integration formula and pulling out the $\tfrac{-1}{1-s}$ term which is the contribution from the eigenvalue at 1. Let $K$ be a positive real and $s\in \mathbb{R}$ such that $0\leq s<1$. Then,
\begin{eqnarray}
&&\int_{SO(2N+1)} \left(\frac{\Lambda_X^{'}}{\Lambda_X}(s)\right)^K \dX\\
&& = \frac{2^{N^2}}{\pi^N N!} \int_{[0,\pi]^N} \left( \frac{-1}{1-s} + \sum_{m=1}^N \frac{(2s-2\cos\theta_m)}{1+s^2-2s\cos\theta_m} \right)^K \Delta^2(\cos\theta) \prod_{m=1}^N \sin^2(\theta_m/2) \dx{\theta}. \label{eqa1}
\end{eqnarray}
We factor out $\left(\frac{-1}{1-s}\right)^K$ and now focus on the integral
\begin{equation}
\int_{[0, \pi]^N}\left(1 + (s-1)\sum_{m=1}^N \frac{2s-2\cos\theta_m}{1+s^2 - 2s\cos\theta_m}\right)^K \Delta^2(\cos\theta) \prod_{m=1}^N \sin^2(\theta_m/2)d\theta.
\end{equation}

In the following Lemma, we show that we can approximate the multinomial sum inside the integral by a polynomial. 
\begin{lemma}\label{taylorlem}
Let $m \in \N, K \in \R$ with $\lfloor K \rfloor \leq m$. Then, for all $\frac{-1}{2} \leq x$,

\begin{equation}
    \left|(1+x)^K - \sum_{n=0}^{m} \binom{K}{n} x^n \right| \leq C_{K,m} |x|^{m+1},
\end{equation}

where $C_{K,m}>0$ is a constant which depends only on $K$ and $m$.
\end{lemma}

\begin{proof}
Let $g(x) = (1+x)^K$. By Taylor's theorem, we can truncate the multinomial expansion of $g(x)$ into a sum and an integral remainder:
\begin{equation}
g(x) = \sum_{n=0}^{m} \binom{K}{n} x^n + \int_0^{x} \frac{g^{(m+1)}(t)(x-t)^{m}}{m !} dt.
\end{equation}

Next, we can bound the integral remainder using Taylor's inequality which tells us that for $\lfloor K \rfloor \leq m$, 
\begin{equation}
\left|g(x) - \sum_{n=0}^{m} \binom{K}{n} x^n\right| \leq ||g^{(m+1)}||_{\infty, [-1/2, +\infty)} \frac{|x|^{m+1}}{(m+1)!},
\end{equation}
where 

\begin{equation}
\frac{||g^{(m+1)}||_{\infty, [-1/2, +\infty)}}{(m+1)!} = \left|\binom{K}{m+1}\right| \frac{1}{(1 - 1/2)^{m+1 - K}}.
\end{equation}

We note that the condition $\lfloor K \rfloor \leq m$, ensures that $g^{(m+1)}$ is bounded on $[-1/2, +\infty)$.

Therefore, let

\begin{equation}
C_{K, m} =\left|\binom{K}{m+1}\right| 2^{m+1 - K},
\end{equation} and we have proven Lemma \ref{taylorlem}.
\end{proof}

The next Lemma allows us to bound the sum in the logarithmic derivative, so we can apply the previous lemma to our moment calculation. It relies on the theory of homographic functions. 

\begin{lemma}\label{boundinglemma}
Let $0 \leq s<1$. Then,
\begin{equation}
\frac{2(s-1)}{s+1}N\leq (s-1)\sum_{j=1}^N \frac{2s-2\cos\theta_j}{1+s^2 - 2s\cos\theta_j}\leq 2N\label{2ineq}
\end{equation} where $\theta_j \in [0,2\pi]$ and $N \in \N$. 
\end{lemma}

\begin{proof} For $0 \leq s < 1$, define
\begin{equation}
f(x) := \frac{2s - 2x}{1+s^2 - 2sx}
\end{equation} for $x\in [-1, 1]$.

Note that $1+s^2 -2sx\geq 1+s^2 - 2s = (1 - s)^2$, therefore for $s <1$,  $f(x)$ is well defined and is a homographic function as in \eqref{homographic} with $a = -2, b= 2s, c = -2s$ and $d = 1 + s^2$. The singularity at $\tfrac{-d}{c}$ is given by $\frac{1+s^2}{2s} > 1$ and $ad - bc $ is given by $2(s^2 - 1)< 0$, therefore $f$ is a decreasing homographic function on the interval $[-1,1]$ and so,

\begin{equation}\label{ineq2}
f(1) = \frac{2s-2}{1+s^2 - 2s}\leq f(x)\leq \frac{2s+2}{1+s^2 +2s} = f(-1)
\end{equation}
holds for all $x\in [-1, 1]$.
Equation \eqref{ineq2} implies
\begin{equation}
\frac{2}{s-1}\leq f(x) \leq \frac{2}{s+1}.
\end{equation} Multiplying all terms by $(s-1)$, which is negative given the range of $s$, inverts the inequalities and finally we have that
\begin{equation}
\frac{2(s-1)}{s+1}N\leq (s-1)\sum_{j=1}^N \frac{2s-2\cos(\theta_j)}{1+s^2 - 2s\cos(\theta_j)}\leq 2N.
\end{equation}
\end{proof}

The following proposition simply applies the two lemmas above to our moment computation. 

\begin{prop}\label{propbound}
For $K \in \mathbb{R}_{+}$, pick $0 < s<1$ and $m$ such that $\frac{-1}{2} \leq \frac{2(s-1)}{s+1}N$ and $\lfloor K \rfloor \leq m$. For simplicity, let $m+1$ be even as well. Then,

\begin{equation}
\abs{(s-1)^K\int_{SO(2N+1)} \left(\frac{\Lambda_X'}{\Lambda_X}(s)\right)^K \dX -\left[\sum_{n = 0}^{m} \binom{K}{n}h(n, s, N)\right]}\leq C_{K,m} h(m+1, s, N)
\end{equation}

where 

\begin{equation}\label{hfunc}
h(n, s, N) :=\frac{2^{N^2}}{\pi^N N!}\int_{[0, \pi]^N}\left((s-1)\sum_{j=1}^N \frac{2s-2\cos(\theta_j)}{1+s^2 - 2s\cos(\theta_j)}\right)^n \Delta(\cos(\theta_j))^2 \prod_{i=1}^N \sin^2(\theta_j/2)d\theta
\end{equation} with $n\in \N$, and
\begin{equation}
    C_{K, m} =\left|\binom{K}{m+1}\right| 2^{m+1 - k}.
\end{equation}

\end{prop}

\begin{proof}
\begin{eqnarray}
&& \abs{(s-1)^K\int_{SO(2N+1)} \left(\frac{\Lambda_X'}{\Lambda_X}(s)\right)^K \dX -\left[\sum_{n = 0}^{m} \binom{K}{n}h(n, s, N)\right]} \\
&& = \left| \frac{2^{N^2}}{\pi^N N!}\int_{[0, \pi]^N}\left(1 + (s-1)\sum_{m=1}^N \frac{2s-2\cos(\theta_m)}{1+s^2 - 2s\cos(\theta_m)}\right)^K \Delta(\cos(\theta_m))^2 \prod_{m=1}^N \sin^2(\theta_m/2)d\theta \right. \\
&&  \left.- \sum_{n=0}^m \binom{K}{n} \frac{2^{N^2}}{\pi^N N!}\int_{[0, \pi]^N}\left((s-1)\sum_{j=1}^N \frac{2s-2\cos(\theta_j)}{1+s^2 - 2s\cos(\theta_j)}\right)^n \Delta(\cos(\theta_j))^2 \prod_{i=1}^N \sin^2(\theta_j/2)d\theta \right| \nonumber \\
&& =  \left|\frac{2^{N^2}}{\pi^N N!}\int_{[0, \pi]^N}\left[\left(1 + (s-1)\sum_{m=1}^N \frac{2s-2\cos(\theta_m)}{1+s^2 - 2s\cos(\theta_m)}\right)^K \nonumber \right. \right. \\
&& \left. \left. - \sum_{n=0}^m \binom{K}{n}\left((s-1)\sum_{j=1}^N \frac{2s-2\cos(\theta_j)}{1+s^2 - 2s\cos(\theta_j)}\right)^n \right] \Delta(\cos(\theta_j))^2 \prod_{i=1}^N \sin^2(\theta_j/2)d\theta \right| \nonumber \\
&& \leq \frac{2^{N^2}}{\pi^N N!}\int_{[0, \pi]^N}\left|\left(1 + (s-1)\sum_{m=1}^N \frac{2s-2\cos(\theta_m)}{1+s^2 - 2s\cos(\theta_m)}\right)^K \right. \nonumber \\
&& \left. - \sum_{n=0}^m \binom{K}{n}\left((s-1)\sum_{j=1}^N \frac{2s-2\cos(\theta_j)}{1+s^2 - 2s\cos(\theta_j)}\right)^n \right| \Delta(\cos(\theta_j))^2 \prod_{i=1}^N \sin^2(\theta_j/2)d\theta.    \label{prop1ineq}
\end{eqnarray} 

Now we apply Lemma \ref{taylorlem} with 
$$ x = (s-1)\sum_{m=1}^N \frac{2s-2\cos(\theta_m)}{1+s^2 - 2s\cos(\theta_m)} $$ which yields the result. Note that since $m+1$ is even, we may drop the absolute value on $x$ on the right hand side. 
\end{proof}

The calculations to this point have been evaluating the logarithmic derivative at a point $s$; if we wish to let $s$ depend on $N$, as in \cite{AS}, we may take $s = e^{-\alpha} = e^{-a/N}$ where $a = o(1)$ as $N$ goes to infinity with $a>0$, and compute the asymptotic behaviour of $h(n, s, N)$ for each fixed $n$.

The following lemma allows us to rewrite $h(n,s,N)$ using a polynomial change of basis, which allows us to use the asymptotic expansion of Theorem \ref{thm4}. 

\begin{lemma}\label{polyexp}
Consider a polynomial $X^n$ expanded in the basis $((1 + X)^k)_{k\leq n}$ with coefficients $a_{k, n}$:

\begin{equation}
X^n = \sum_{k=0}^n a_{k, n} (1+X)^k.
\label{eq2}
\end{equation} Then,
\begin{equation}
    \sum_{k=0}^n a_{k, n} =  \delta_{0, n} \label{delta0}
\end{equation} and 
\begin{equation}
 \sum_{k=0}^n a_{k, n} k = \delta_{1, n} \label{delta1}.
\end{equation}
\begin{proof}
Equation \eqref{delta0} follows simply by taking $X = 0$, and \eqref{delta1} follows by differentiating with respect to $X$, then setting $X = 0$.
\end{proof}
\end{lemma}

Now, the condition on Proposition \eqref{propbound} we need to satisfy is $\frac{-1}{2} \leq \frac{2(s-1)}{(s+1)}N = \frac{2N(e^{-a/N} - 1)}{e^{-a/N} +1} \sim \frac{-2aN}{2N-a} \sim -a$ which is met if $a=o(1)$.

\begin{prop}\label{propintegercase}
With $h(n, s, N)$ defined as in \eqref{hfunc} and $a = o(1)$ as $N \to \infty$ as before,

\begin{equation}
h(n, e^{-a/N}, N) = \left\{
    \begin{array}{ll}
        1 + \mathcal{O}(a/N) + \mathcal{O}(a^2)& \mbox{if } n =0 \\
        -a + \mathcal{O}(a/N) + \mathcal{O}(a^2) & \mbox{if } n =1 \\
        \mathcal{O}(a/N) + \mathcal{O}(a^2) & \mbox{otherwise.}
    \end{array}
\right.
\end{equation}
\end{prop}

\begin{proof}

Factoring out $\left(\frac{1}{e^{-a/N}-1}\right)^n \sim \left( \frac{-N}{a}\right)^n$ from Theorem \ref{thm4}, we have

\begin{eqnarray}
&&\frac{2^{N^2}}{\pi^N N!}\int_{[0, \pi]^N}\left(1+(e^{-a/N}-1)\sum_{j=1}^N \frac{2e^{-a/N}-2\cos(\theta_j)}{1+e^{-2a/N} - 2e^{-a/N}\cos(\theta_j)}\right)^n \Delta(\cos(\theta_j))^2 \prod_{i=1}^N \sin^2(\theta_j/2)d\theta \nonumber \\
&&= 1 -n a + \mathcal{O}(a/N) + \mathcal{O}(a^2). \label{asymexp}
\end{eqnarray}

Therefore, applying Lemma \ref{polyexp}, $h(n, e^{-a/N}, N)$ equals

\begin{eqnarray}
&& \quad  \quad  \frac{2^{N^2}}{\pi^N N!}\int_{[0, \pi]^N}\left((e^{-a/N}-1)\sum_{j=1}^N \frac{2e^{-a/N}-2\cos(\theta_j)}{1+e^{-2a/N} - 2e^{-a/N}\cos(\theta_j)}\right)^n \Delta(\cos(\theta_j))^2 \prod_{i=1}^N \sin^2(\theta_j/2)d\theta   \label{prop2sum}  \\
&& =\frac{2^{N^2}}{\pi^N N!}\int_{[0, \pi]^N} \sum_{k=0}^n a_{k, n} \left(1 + (e^{-a/N}-1)\sum_{j=1}^N \frac{2e^{-a/N}-2\cos(\theta_j)}{1+e^{-2a/N} - 2e^{-a/N}\cos(\theta_j)}\right)^k \Delta(\cos(\theta_j))^2 \prod_{i=1}^N \sin^2(\theta_j/2)d\theta\nonumber\\
&& =\sum_{k=0}^n a_{k, n} \frac{2^{N^2}}{\pi^N N!}\int_{[0, \pi]^N}  \left(1 + (e^{-a/N}-1)\sum_{j=1}^N \frac{2e^{-a/N}-2\cos(\theta_j)}{1+e^{-2a/N} - 2e^{-a/N}\cos(\theta_j)}\right)^k \Delta(\cos(\theta_j))^2 \prod_{i=1}^N \sin^2(\theta_j/2)d\theta. \nonumber 
\end{eqnarray}
Now, taking the asymptotic expansion as in \eqref{asymexp}, \eqref{prop2sum} equals
\begin{eqnarray}
&& \sum_{k=0}^n a_{k, n} \left(1 -k a + \mathcal{O}(a/N) + \mathcal{O}(a^2)\right)\nonumber\\
&& =\sum_{k=0}^n a_{k, n} -a\sum_{k=0}^n a_{k, n} k + \mathcal{O}(a/N) + \mathcal{O}(a^2) \\
&& = \left\{
    \begin{array}{ll}
        1 + \mathcal{O}(a/N) + \mathcal{O}(a^2)& \mbox{if } n =0 \\
        -a + \mathcal{O}(a/N) + \mathcal{O}(a^2) & \mbox{if } n =1 \\
        \mathcal{O}(a/N) + \mathcal{O}(a^2) & \mbox{otherwise.}
    \end{array}
\right.
\end{eqnarray}
\end{proof}

We are now ready to prove the main result. 
\begin{customthm}{2}
Let $\Lambda_X(s)$ denote the characteristic polynomial of a matrix $X \in SO(2N+1)$, the group of odd orthogonal random matrices of odd dimension and determinant 1, equipped with the Haar measure $\dX$. Let $\alpha = a/N>0$ with $a \in \mathbb{R}_{+}$, $a = o(1)$ as $N \to \infty$, and the moment parameter $K \in \mathbb{R}_{+}$. Then, 
\begin{equation}
\left(\frac{-N}{a}\right)^K\int_{SO(2N+1)} \left(\frac{\Lambda_X'}{\Lambda_X}(e^{-\alpha})\right)^K \dX = 1 - Ka + \mathcal{O}(a/N) + \mathcal{O}(a^2).
\end{equation}
\end{customthm}

\begin{proof}
Fix an integer $m$ which satisfies three things: $m\geq \lfloor K \rfloor$, $m$ is odd and $m\geq 2$. By the last condition, we know from Proposition \ref{propintegercase}, that
\begin{equation}
    h(m+1, e^{-a/N}, N) = \mathcal{O}(a/N) + \mathcal{O}(a^2).
\end{equation}

By Proposition \ref{propbound}, we know that  

\begin{eqnarray}
&& (e^{-a/N}-1)^K\int_{SO(2N+1)} \left(\frac{\Lambda_X'}{\Lambda_X}(e^{-a/N})\right)^K \dX \nonumber \\ && = \sum_{n = 0}^{m} \binom{K}{n}h(n, e^{-a/N}, N) + \mathcal{O}\left( h(m+1, e^{-a/N}, N)\right) \\
&& = h(0,e^{-a/N}, N) + Kh(1, e^{-a/N}, N) + \mathcal{O}(a/N) + \mathcal{O}(a^2) \nonumber \\
&& = 1 - Ka + \mathcal{O}(a/N) + \mathcal{O}(a^2),
\end{eqnarray}
where the last line comes from applying Proposition \ref{propintegercase} again and 
\begin{equation}
   (e^{-a/N}-1)^K \sim \left(\frac{-N}{a}\right)^K 
\end{equation} as $N \to \infty$.
\end{proof}

\section{Proof of Theorem \ref{theorem2}}\label{last_section}

Let $f$ and $g$ be some functions. In this section, we will write $f \ll g$ or $f = \mathcal{O}(g)$ if there exists a constant $A>0$, independent of any variables like $s$, $N$ or $\theta$, such that $f\leq A g$. In the following we have many terms with size $N^p(1-s)^q$ for various $p$ and $q$. Note that since $s=e^{-a/N}$, for large $N$ we find that $1-s$ behaves like $a/N$, where we remember that $a$ goes to zero as $N$ goes to infinity. It is not always possible to determine which term dominates when $N$ is large until we impose a constraint on how fast $a$ goes to zero as $N\rightarrow \infty$. We will impose such a constraint at the end, to achieve a clean result in Theorem 3, but will carry the various terms through the calculation as it gives the intermediate results versatility.

We will use the $1-$ and $2$-point correlation functions of eigenvalues of matrices from $SO(2N+1)$. We recall the expression of the $n$-point correlation functions $R_n$:

\begin{equation}
R_n(\theta_1, ..., \theta_n) = \underset{n\times n}{\det}(S_{2N}(\theta_k - \theta_j) - S_{2N}(\theta_k + \theta_j)),
\end{equation}
where 

\begin{equation}
S_N(\theta) = \frac{1}{2\pi} \frac{\sin(N\theta/2)}{\sin(\theta/2)}.
\end{equation}

In particular, we have that

\begin{equation}
    R_1(\theta):=S_{2N}(0) - S_{2N}(2\theta) = \frac{N}{\pi} - \frac{\sin(2N\theta)}{2\pi \sin(\theta)}
\end{equation}
and
\begin{eqnarray}
R_2(\theta_1, \theta_2) &&= \begin{vmatrix}
S_{2N}(0) - S_{2N}(2\theta_1) && S_{2N}(\theta_1 - \theta_2) - S_{2N}(\theta_1 + \theta_2)\\
S_{2N}(\theta_1 - \theta_2) - S_{2N}(\theta_1 + \theta_2) && S_{2N}(0) - S_{2N}(2\theta_2)\\
\end{vmatrix}\\
&& = R_1(\theta_1)R_1(\theta_2) - (S_{2N}(\theta_1 - \theta_2) - S_{2N}(\theta_1 + \theta_2))^2. \label{R_2}
\end{eqnarray}

Finally, we recall the formula for averaging a function of the eigenvalues over $SO(2N+1)$, which explains why these functions will be useful:

\begin{equation}
    \int_{[0, \pi]^N} f(\theta_1, ..., \theta_n)d\mu(\theta) = \frac{(N-n)!}{N!}\int_{[0, \pi]^n} f(\theta_1, ..., \theta_n)R_n(\theta_1, ..., \theta_n)d\theta_1...d\theta_n,
\end{equation}
where $d\mu(\theta) = d\mu(\theta_1, ..., \theta_n) := \frac{2^{N^2}}{\pi^N N!}\Delta^2(\cos(\theta))\prod_{m=1}^N \sin^2(\theta_m/2)d\theta_1...d\theta_n$ and $f$ is a function of $n$ variables with $1\leq n\leq N$.

\bs
We recall the function $h(n, s, N)$, previously defined by

\begin{equation}\label{hint}
    h(n, s, N) = \int_{[0, \pi]^N}\left((s-1)\sum_{j=1}^N \frac{2s-2\cos(\theta_j)}{1+s^2 - 2s\cos(\theta_j)}\right)^n d\mu(\theta),
\end{equation}
where $n\in \N$. $h(n, s, N)$ is the key quantity to study asymptotically. This is our next goal.

We note that

\begin{equation}
\frac{2s-2\cos(\theta_j)}{1+s^2 - 2s\cos(\theta_j)} = \frac{1}{s}\left(1+  \frac{s^2-1}{1+s^2 - 2s\cos(\theta_j)}\right) = \frac{1}{s}\left(1+  \frac{s^2-1}{(1-s)^2 + 4s\sin^2(\theta_j/2)}\right).
\end{equation}

Then, we can rewrite \eqref{hint} as 

\begin{equation}
\frac{(s-1)^n}{s^n}\int_{[0, \pi]^N}\left(N + \sum_{j=1}^N \frac{s^2-1}{(1-s)^2 + 4s\sin^2(\theta_j/2)}\right)^n d\mu(\theta).
\end{equation}

We can expand the expression with the multinomial theorem to write

\begin{equation}
\frac{(s-1)^n}{s^n}\int_{[0, \pi]^N}\sum_{p + l_1 + ... + l_N = n}\binom{n}{p, l_1, ..., l_N}N^{p} \\\prod_{j=1}^N \left(\frac{s^2-1}{(1-s)^2 + 4s\sin^2(\theta_j/2)}\right)^{l_j} d\mu(\theta),
\end{equation}
where 

\begin{equation}
\binom{n}{p, l_1, ..., l_N} = \frac{n!}{p! l_1!...l_N!}.
\end{equation}

Finally, 

\begin{multline}\label{h}
h(n, s, N) = \frac{(s-1)^n}{s^n}\sum_{p + l_1 + ... + l_N = n}\binom{n}{p, l_1, ..., l_N}N^{p} (s^2-1)^{n-p} \\\int_{[0, \pi]^N}\prod_{j=1}^N \frac{1}{((1-s)^2 + 4s\sin^2(\theta_j/2))^{l_j}} d\mu(\theta).
\end{multline}

So, we have to understand the asymptotic behaviour of 

\begin{equation}\label{g}
    g(l, s, N) :=\\\int_{[0, \pi]^N}\prod_{j=1}^N \frac{1}{((1-s)^2 + 4s\sin^2(\theta_j/2))^{l_j}} d\mu(\theta),
\end{equation}
with $l = (l_1, ..., l_N)$ and $l_1 + ... + l_N = n - p$.

We will begin by studying $g$ in two cases (which can overlap):
\begin{itemize}
    \item at least one of the $l_j$'s is bigger than or equal to $2$ or
    \item at least two the $l_j$'s are equal to $1$.
\end{itemize}

\subsection{First case}

Here, we will assume that at least one of the $l_j$'s is bigger than or equal to $2$. Without loss of generality, we can assume that $l_1\geq 2$ since $g$ is invariant under a permutation of the components of $l$.

For all $j\geq 2$, we use this inequality

\begin{equation}
[(1-s)^2 + 4s\sin^2(\theta_j/2)]^{l_j} \geq [1-s]^{2l_j}
\end{equation}
in order to state that

\begin{equation}
g(l, s, N)\leq \frac{1}{(1-s)^{2\sum_{j\geq 2}l_j}} \int_{[0, \pi]^N} \frac{1}{((1-s)^2 + 4s\sin^2(\theta_1/2))^{l_1}} d\mu(\theta).
\end{equation}

Hence, 
\begin{equation}
g(l, s, N)\leq \frac{1}{(1-s)^{2n - 2(p + l_1)}} \int_{[0, \pi]^N} \frac{1}{((1-s)^2 + 4s\sin^2(\theta_1/2))^{l_1}} d\mu(\theta)
\end{equation}
because $\sum_{j\geq 2}l_j = n - p - l_1$.

Using the $1$-point correlation function, which is positive, we can rewrite this above integral and so we have that

\begin{equation}\label{inequality_g}
    g(l, s, N)\leq \frac{1}{N(1-s)^{2n - 2(p + l_1)}}\int_{0}^{\pi} \frac{R_1(\theta)}{((1-s)^2 + 4s\sin^2(\theta/2))^{l_1}}d\theta.
\end{equation}

We want to bound $R_1$ by something which keeps track of the behaviour of $R_1$ close to $0$. In other words, we want to keep in mind the fact that the $\theta_n$'s are repelled by $1$.

\begin{lemma}\label{boundingR_1}

Let $\theta\in [-\pi, \pi]$. We have that

\begin{equation}
    R_1(\theta) = \left(\frac{N}{\pi} - \frac{\sin(2N\theta)}{2\pi \sin(\theta)}\right)\ll N^3 \sin^2(\theta/2),
\end{equation}
where the constant behind $\ll$ doesn't depend on $\theta$ and on $N$.

\end{lemma}

\begin{proof}
Let's begin by noticing that 
\begin{equation}
    \sum_{k=0}^{N-1}e^{ik\theta} = \frac{\sin(N\theta/2)}{\sin(\theta/2)}e^{i\frac{\theta}{2}(N - 1)}
\end{equation}
and so

\begin{equation}
    \sum_{k=0}^{2N-1}e^{i2k\theta} = \frac{\sin(2N\theta)}{\sin(\theta)}e^{i\theta(2N - 1)}.
\end{equation}

Hence, we can rewrite the $1$-point correlation function in the following way:

\begin{equation}\label{R_1}
    R_1(\theta):=\frac{N}{\pi} - \frac{\sin(2N\theta)}{2\pi \sin(\theta)} =  \frac{N}{\pi} - \frac{1}{2\pi}\sum_{k=0}^{2N-1} e^{i\theta(2k+1-2N)}.
\end{equation}

Once again, we use Taylor's formula:

\begin{equation}
    R_1(\theta) = R_1(0) + R_1'(0)\theta + \int_{0}^{\theta} R_1''(t)tdt,
\end{equation}

and $R_1(0) = 0$ and $R_1'(0) = 0$. Indeed,

\begin{equation}
R_1'(0) = -\frac{i}{2\pi}\sum_{k=0}^{2N-1} (2k+1-2N) = -\frac{i}{2\pi}\left(2\sum_{k=0}^{2N-1}k + (1-2N)2N\right) = 0.
\end{equation}

So, we have

\begin{equation}
    R_1(\theta) = \int_{0}^{\theta} R_1''(t)(\theta-t)tdt.
\end{equation}

We notice that

\begin{equation}
    |R_1''(t)| = \left|\frac{1}{2\pi}\sum_{k=0}^{2N-1} (2k+1-2N)^2e^{it(2k+1-2N)}\right|\leq \frac{1}{2\pi}\sum_{k=0}^{2N-1} (2k+1-2N)^2.
\end{equation}

So, we can find $M>0$ such as $|R_1''(t)|\leq M N^3$ for all $t\in [0, \pi]$.

Finally, we have proved that

\begin{equation}
    |R_1(\theta)|\leq M N^3 \int_{0}^{\theta}tdt = \frac{M}{2} N^3 \theta^2.
\end{equation}

Then, we can find $C>0$ such as $\frac{M}{2}N^3 \theta^2\leq C N^3\sin^2(\theta/2)$. Indeed, the function $\theta\mapsto \frac{M \theta^2}{2 C \sin^2(\theta/2)}$ can be extended by continuity at $0$, therefore it is continuous on the compact $[-\pi, \pi]$ and so is bounded.

\end{proof}

\bs

Let's come back to inequality (\ref{inequality_g}). By using this lemma, we have proved that

\begin{equation}\label{inequality_g2}
g(l, s, N)\ll \frac{N^2}{(1-s)^{2n-2(p + l_1)}} \int_{0}^{\pi} \frac{\sin^2(\theta/2)}{[(1-s)^2 + 4s\sin^2(\theta/2)]^{l_1}}d\theta.
\end{equation}

And so, we need to study 

\begin{equation}\label{i1}
    I_{l_1} := \int_{0}^{\pi} \frac{\sin^2(\theta/2)}{[(1-s)^2 + 4s\sin^2(\theta/2)]^{l_1}}d\theta.
\end{equation}

We will apply the change of variable $t = \tan(\theta/2)$, and recall some relevant formulas:

\begin{equation}
\cos(\theta) = \frac{1-t^2}{1+t^2} \text{ and } d\theta = \frac{2}{1+t^2}dt
\end{equation}
\begin{equation}
\sin^2(\theta/2) = \frac{1}{2}(1 - \cos(\theta)) = \frac{1}{2}\left(1 - \frac{1-t^2}{1+t^2}\right) = \frac{t^2}{1+t^2}
\end{equation}
\begin{equation}
(1-s)^2 + 4s\sin^2(\theta/2) = 1+s^2 - 2s\cos(\theta).
\end{equation}

Now, with the change of variable, equation \eqref{i1} becomes

\begin{eqnarray}
I_{l_1} && = \int_{0}^{+\infty} \frac{t^2/(1+t^2)}{(1+s^2 - 2s\frac{1-t^2}{1+t^2})^{l_1}} \frac{2}{1+t^2}dt \nonumber\\
 &&= 2\int_{0}^{+\infty} \frac{(1+t^2)^{l_1-2}t^2}{((1+s^2)(1+t^2) - 2s(1-t^2))^{l_1}}dt \nonumber\\
 && =2\int_{0}^{+\infty} \frac{(1+t^2)^{l_1-2}t^2}{[t^2(1+s)^2 + (1-s)^2]^{l_1}}dt.
\end{eqnarray}

Now, applying another change of variable $u = \frac{1+s}{1-s}t$, we obtain

\begin{equation}
    I_{l_1} = \frac{2}{(1-s)^{2l_1-3}(1+s)^3} \int_0^{+\infty} \frac{(1+\left(\frac{1-s}{1+s}u\right)^2)^{l_1-2} u^2}{(1+u^2)^{l_1}}du.
\end{equation}

As $l_1\geq 2$, we have that

\begin{eqnarray}
I_{l_1} &&\leq \frac{2}{(1-s)^{2l_1-3}(1+s)^3} \int_0^{+\infty} \frac{(1 + u^2)^{l_1 - 2}u^2}{(1+u^2)^{l_1}}du \nonumber\\
&& \leq \frac{2}{(1-s)^{2l_1-3}(1+s)^3} \int_0^{+\infty} \frac{u^2}{(1+u^2)^{2}}du \nonumber \\
&& \leq \frac{1}{(1-s)^{2l_1-3}(1+s)^3} \int_0^{+\infty} \frac{u^2}{1+u^2}\frac{2du}{1 + u^2} \nonumber\\
&& \leq \frac{1}{(1-s)^{2l_1-3}(1+s)^3} \int_0^{\pi} \sin^2(\theta/2)d\theta \nonumber\\
&& \leq \frac{\pi}{2(1-s)^{2l_1-3}(1+s)^3} \nonumber\\
&& \ll \frac{1}{(1-s)^{2l_1-3}}. \label{inequality_I1}
\end{eqnarray}

We can now state the main theorem of this section.

\begin{theorem}

If at least one of the $l_j$'s is bigger than or equal to $2$. We have that

\begin{equation}
    g(l, s, N)\ll N^2 \frac{(1-s)^3}{(1-s)^{2n - 2p}}.
\end{equation}
\end{theorem}

\begin{proof}
We just have to use both inequalities (\ref{inequality_g2}) and (\ref{inequality_I1}).
\end{proof}

\subsection{Second case}

Let's come back to \eqref{g}. We will now assume that at least two of the $l_j$'s are equal to $1$. We can assume that $l_1 = l_2 = 1$.

For all $j\geq 3$, we use this inequality

\begin{equation}
[(1-s)^2 + 4s\sin^2(\theta_j/2)]^{l_j} \geq [1-s]^{2l_j}
\end{equation}
in order to state that

\begin{equation}
g(l, s, N)\leq \frac{1}{(1-s)^{2\sum_{j\geq 3}l_j}} \int_{[0, \pi]^N} \prod_{j=1}^2 \frac{1}{((1-s)^2 + 4s\sin^2(\theta_j/2))^{l_j}} d\mu(\theta).
\end{equation}

But, we know that $\sum_{j\geq 3}l_j = n - (p + l_1 + l_2) = n- p - 2$ and so, 

\begin{equation}
g(l, s, N)\leq \frac{1}{(1-s)^{2n - 2p - 4}} \int_{[0, \pi]^N} \prod_{j=1}^2\frac{1}{((1-s)^2 + 4s\sin^2(\theta_j/2))^{l_j}} d\mu(\theta).
\end{equation}

Using the pair correlation function which is positive, we can write this integral as the following way:

\begin{equation}
    \frac{1}{(1-s)^{2n - 2p - 4}}\frac{1}{N(N-1)}\int_{0}^{\pi}\int_{0}^{\pi} \prod_{j=1}^2 \frac{1}{(1-s)^2 + 4s\sin^2(\theta_j/2)}R_2(\theta_1, \theta_2)d\theta_1d\theta_2.
\end{equation}

Hence, we obtain that

\begin{multline}
    g(l, s, N)\leq \frac{1}{N(N-1)(1-s)^{2n - 2p - 4}}\left[\left(\int_{0}^{\pi} \frac{1}{(1-s)^2 + 4s\sin^2(\theta/2)} R_1(\theta)d\theta\right)^2\right.\\
    -\left.\int_{0}^{\pi}\int_{0}^{\pi} \prod_{j=1}^2 \frac{1}{(1-s)^2 + 4s\sin^2(\theta_j/2)}(S_{2N}(\theta_1 - \theta_2) - S_{2N}(\theta_1 + \theta_2))^2d\theta_1d\theta_2\right]
\end{multline}
because $R_2(\theta_1,\theta_2)=R_1(\theta_1)R_1(\theta_2) - (S_{2N}(\theta_1 - \theta_2) - S_{2N}(\theta_1 + \theta_2))^2$ from (\ref{R_2}).

But, by using $R_1(\theta)\ll  N^3 \sin^2(\theta/2)$ from Lemma \eqref{boundingR_1}, we have that

\begin{eqnarray}
\label{integralbounded} \int_{0}^{\pi} \frac{1}{(1-s)^2 + 4s\sin^2(\theta/2)} R_1(\theta)d\theta &&\ll N^3 \int_{0}^{\pi} \frac{\sin^2(\theta/2)}{(1-s)^2 + 4s\sin^2(\theta/2)} d\theta  \\
&& \ll N^3 \int_{0}^{\pi} \frac{\sin^2(\theta/2)}{4s\sin^2(\theta/2)} d\theta \nonumber \\
&& \ll N^3 \int_{0}^{\pi} \frac{1}{4s}d\theta \nonumber\\
&& \ll N^3.
\end{eqnarray}

So, we obtain that

\begin{multline}\label{inequality_g3}
    g(l, s, N)\ll \frac{1}{N(N-1)(1-s)^{2n - 2p - 4}}\Big[N^6 +  \\
    \int_{0}^{\pi}\int_{0}^{\pi} \prod_{j=1}^2 \frac{1}{(1-s)^2 + 4s\sin^2(\theta_j/2)}(S_{2N}(\theta_1 - \theta_2) - S_{2N}(\theta_1 + \theta_2))^2d\theta_1d\theta_2\Big].
\end{multline}

We know that $S_{2N}(x) = \frac{N}{\pi} + \mathcal{O}(N^3 x^2)$ using Lemma \ref{boundingR_1}. However, this estimation provides the term $N^6$ which is too big for what we want. We will use the following lemma instead.

\begin{lemma}
Let $x\in \R$. We have that

\begin{equation}
S_{2N}(x) = \frac{N}{\pi} + \mathcal{O}(N^2 |x|).
\end{equation}
\end{lemma}

\begin{proof}
We have that
\begin{eqnarray}
\left|S_{2N}(x) - \frac{N}{\pi}\right| &&= |R_1(x/2)| \nonumber \\
&&\leq \left|\int_{0}^{x/2} R_1'(t)dt\right|\\
&&\leq |x/2| \left\|\frac{-i}{2\pi}\sum_{k=0}^{2N-1} (2k+1 - 2N)e^{it (2k+1-2N)}\right\|_{\infty, t} \text{ using (\ref{R_1})}  \nonumber \\
&& \ll |x| \sum_{k=0}^{2N-1} |2k+1-2N| \nonumber \\
&& \ll N^2 |x|.
\end{eqnarray}
\end{proof}

So, we can say that

\begin{eqnarray}
\left[S_{2N}(\theta_1 - \theta_2) - S_{2N}(\theta_1 + \theta_2)\right]^2 &&= \left[\frac{N}{\pi} - \frac{N}{\pi} + \mathcal{O}(N^2 |\theta_1 - \theta_2|) + \mathcal{O}(N^2 |\theta_1 + \theta_2|)\right]^2\\
&&= \left[\mathcal{O}(N^2(|\theta_1 - \theta_2| + |\theta_1 + \theta_2|))\right]^2 \nonumber \\
&& \ll N^4 (|\theta_1-\theta_2| + |\theta_1 + \theta_2|)^2 \nonumber \\
&& \ll N^4(|\theta_1-\theta_2|^2 + |\theta_1 + \theta_2|^2),
\end{eqnarray}
because $(a+b)^2 = a^2 + b^2 +2ab = a^2 +b^2 + (a^2 + b^2 - (a-b)^2) \leq 2(a^2 + b^2)$.

Hence, 

\begin{eqnarray}
\left[S_{2N}(\theta_1 - \theta_2) - S_{2N}(\theta_1 + \theta_2)\right]^2 &&\ll N^4(|\theta_1-\theta_2|^2 + |\theta_1 + \theta_2|^2)\\
&&\ll N^4( \theta_1^2 + \theta_2^2 - 2\theta_1 \theta_2 + \theta_1^2 + \theta_2^2 + 2\theta_1\theta_2) \nonumber \\
&& \ll N^4 (\theta_1^2 + \theta_2^2) \nonumber \\
&& \ll N^4(\sin^2(\theta_1/2) + \sin^2(\theta_2/2)) \label{inequality_S2N}
\end{eqnarray}
where the last line can be justified by noticing that the function $\theta \mapsto \frac{\theta^2}{\sin^2(\theta/2)}$, naturally extended at $0$ by continuity, is bounded on the compact interval $[0,  \pi]$.

Let's use (\ref{inequality_S2N}) in (\ref{inequality_g3}) in order to obtain that

\begin{multline}
    g(l, s, N)\ll \frac{1}{N(N-1)(1-s)^{2n - 2p - 4}}\Big[N^6 +  \\
    N^4\int_{0}^{\pi}\int_{0}^{\pi} \prod_{j=1}^2 \frac{1}{(1-s)^2 + 4s\sin^2(\theta_j/2)}(\sin^2(\theta_1/2) + \sin^2(\theta_2/2))d\theta_1d\theta_2\Big],
\end{multline}
and 
\begin{multline}\label{inequality_g4}
    g(l, s, N)\ll \frac{1}{N(N-1)(1-s)^{2n - 2p - 4}}\left[N^6 +\right. \\
    \left.2N^4\left(\int_{0}^{\pi}\frac{1}{(1-s)^2 + 4s\sin^2(\theta/2)}d\theta\right)\left(\int_{0}^{\pi}\frac{\sin^2(\theta/2)}{(1-s)^2 + 4s\sin^2(\theta/2)}d\theta\right)\right].
\end{multline}

We know that 

\begin{equation}\label{bounded}
\int_{0}^{\pi}\frac{\sin^2(\theta/2)}{(1-s)^2 + 4s\sin^2(\theta/2)}d\theta\leq \int_{0}^{\pi}\frac{\sin^2(\theta/2)}{4s\sin^2(\theta/2)}d\theta = \mathcal{O}(1)
\end{equation}
is bounded.

Moreover, we can compute the other integral by using the following lemma.

\begin{lemma}\label{integralI}
We have that

\begin{equation}
    I(s):=\int_{0}^{\pi} \frac{1}{(1-s)^2 + 4s \sin^2(\theta/2)}d\theta = \frac{\pi}{1-s^2}.
\end{equation}
\end{lemma}

\begin{proof}
We do the change of variable $t = \tan(\theta/2)$ to obtain 

\begin{eqnarray}
I(s) && = \int_{0}^{+\infty} \frac{1}{(1-s)^2 + 4s \frac{t^2}{1 - t^2}} \frac{2}{1 - t^2}dt\\
&& = 2\int_{0}^{+\infty} \frac{1}{(1-s)^2 (1-t^2) + 4st^2}dt \nonumber \\
&&= 2\int_{0}^{+\infty} \frac{1}{(1-s)^2 + (1 + s)^2 t^2}dt \nonumber\\
&& = \frac{2}{(1-s)^2} \int_{0}^{+\infty} \frac{1}{1 + \left(\frac{1+s}{1-s} t\right)^2}dt.
\end{eqnarray}

We do the change of variable $u = \frac{1+s}{1-s} t$ and so,

\begin{eqnarray}
I(s)&& = \frac{2}{1-s^2} \int_{0}^{+\infty} \frac{1}{1 + u^2}du\\
&& = \frac{\pi}{1-s^2}.
\end{eqnarray}

\end{proof}

Let's come back to \eqref{inequality_g4}. By using this lemma and \eqref{bounded}, we have that

\begin{eqnarray}
2N^4\left(\int_{0}^{\pi}\frac{1}{(1-s)^2 + 4s\sin^2(\theta/2)}d\theta\right)\left(\int_{0}^{\pi}\frac{\sin^2(\theta/2)}{(1-s)^2 + 4s\sin^2(\theta/2)}d\theta\right) &&= \mathcal{O}\left(\frac{2\pi N^4}{1-s^2}\right) \nonumber \\
&& = \mathcal{O}\left(\frac{2\pi N^4}{(1-s)(1+s)}\right)\\
&& = \mathcal{O}\left(\frac{N^4}{1-s}\right).
\end{eqnarray}
It is unclear whether this term dominates $N^6$ at this time; to know we would have to set how fast $a \to 0$ as $N \to \infty$. For now, we keep both terms and later we fix the growth of $a$ to find the dominating term.  
Finally, we have that

\begin{equation}
    g(l, s, N) \ll \frac{1}{N(N-1)(1-s)^{2n - 2p - 4}}\left(N^6 + \frac{N^4}{1-s} \right),
\end{equation}
and so we obtain the following theorem.
\begin{theorem}
If at least two of the $l_j$'s are equal to $1$, we have that 
\begin{equation}
    g(l, s, N)\ll \frac{N^4(1-s)^4}{(1-s)^{2n-2p}} + \frac{N^2(1-s)^3}{(1-s)^{2n-2p}}.
\end{equation}

\end{theorem}

\subsection{Back to $h$}

Let's sum up what we have seen. We want to study the quantity $h(n, s, N)$ which is equal, using (\ref{h}), to 

\begin{equation}
    \frac{(s-1)^n}{s^n}\sum_{p + l_1 + ... + l_N = n}\binom{n}{p, l_1, ..., l_N} (s^2-1)^{n-p} g(l, s, N)
\end{equation}
where

\begin{equation}
    g(l, s, N) :=\\\int_{[0, \pi]^N}\prod_{j=1}^N \frac{1}{((1-s)^2 + 4s\sin^2(\theta_j/2))^{l_j}} d\mu(\theta).
\end{equation}

The two last sections proved the following theorem.

\begin{theorem}\label{bothcases}
If at least one of the $l_j$'s is bigger than or equal to $2$ or if at least two of the $l_j$'s are equal to $1$, we have that

\begin{equation}
    g(l, s, N)\ll \frac{N^4(1-s)^4}{(1-s)^{2n-2p}} + \frac{N^2(1-s)^3}{(1-s)^{2n-2p}}.
\end{equation}
\end{theorem}

We are now ready to study the aymptotic of $h$ depending of three cases $n\geq 2$, $n = 1$ and $n = 0$.

\subsection{Case $n\geq 2$}

We can split the sum in three parts according to the value of $p$:

\begin{multline}
    h(n, s, N) = \frac{(s-1)^n}{s^n}\sum_{\underset{p = n}{p + l_1 + ... + l_N = n}}\binom{n}{p, l_1, ..., l_N} N^{p}(s^2-1)^{n-p} g(l, s, N)+\\
    \frac{(s-1)^n}{s^n}\sum_{\underset{p = n-1}{p + l_1 + ... + l_N = n}}\binom{n}{p, l_1, ..., l_N} N^{p}(s^2-1)^{n-p} g(l, s, N)+\\
    \frac{(s-1)^n}{s^n}\sum_{\underset{p\leq n-2}{p + l_1 + ... + l_N = n}}\binom{n}{p, l_1, ..., l_N} N^{p}(s^2-1)^{n-p} g(l, s, N).
\end{multline}

which is equal to

\begin{multline}
    \frac{(s-1)^n}{s^n}N^n+\\
    \frac{(s-1)^n}{s^n}n(s^2-1) N^{n-1}\int_{[0, \pi]^N} \left(\sum_{j=1}^{N}\frac{1}{(1-s)^2 + 4s\sin^2(\theta_j/2)}\right)d\mu(\theta)+\\
    \frac{(s-1)^n}{s^n}\sum_{\underset{p\leq n-2}{p + l_1 + ... + l_N = n}}\binom{n}{p, l_1, ..., l_N} N^{p}(s^2-1)^{n-p} g(l, s, N).
\end{multline}

where 
\begin{eqnarray}
\int_{[0, \pi]^N} \left(\sum_{j=1}^{N}\frac{1}{(1-s)^2 + 4s\sin^2(\theta_j/2)}\right)d\mu(\theta)&& = \int_{0}^{\pi} \frac{R_1(\theta)}{(1-s)^2 + 4s\sin^2(\theta/2)}d\theta \nonumber \\
&&\ll N^3 \int_0^{\pi} \frac{\sin^2(\theta/2)}{(1-s)^2 + 4s\sin^2(\theta/2)}d\theta \quad \text{ (using Lemma \eqref{boundinglemma})}  \\
&&\ll N^3 \int_0^{\pi} \frac{\sin^2(\theta/2)}{4s\sin^2(\theta/2)}d\theta \nonumber\\
&&\ll N^3.
\end{eqnarray}

So, we have that

\begin{eqnarray}\label{h2}
     h(n, s, N) = && \frac{(s-1)^n}{s^n}N^n + \mathcal{O}(N^{n+2} (1-s)^{n+1}) \nonumber \\ && + \frac{(s-1)^n}{s^n}\sum_{\underset{p\leq n-2}{p + l_1 + ... + l_N = n}}\binom{n}{p, l_1, ..., l_N} N^{p}(s^2-1)^{n-p} g(l, s, N).
\end{eqnarray}

Furthermore, if $p\leq n-2$, we have that $l_1 + ... + l_N\geq 2$ since $p + ... + l_N = n$ and so, either at least one of the $l_j$'s is bigger than or equal to $2$ or at least two of the $l_j$'s are equal to $1$. We can apply the Theorem \ref{bothcases} and so, 

\begin{multline}
\left|\frac{(s-1)^n}{s^n}\sum_{\underset{p\leq n-2}{p + l_1 + ... + l_N = n}}\binom{n}{p, l_1, ..., l_N} N^{p}(s^2-1)^{n-p} g(l, s, N)\right|\ll\\
\frac{(s-1)^n}{s^n}\sum_{\underset{p\leq n-2}{p + l_1 + ... + l_N = n}}\binom{n}{p, l_1, ..., l_N} N^{p}|(s^2-1)|^{n-p} \left(\frac{N^4(1-s)^4}{(1-s)^{2n-2p}} + \frac{N^2(1-s)^3}{(1-s)^{2n-2p}} \right)
\end{multline}
\begin{eqnarray}
&& \ll \left( N^4 (1-s)^4 + N^2 (1-s)^3 \right)\sum_{\underset{p\leq n-2}{p + l_1 + ... + l_N = n}}\binom{n}{p, l_1, ..., l_N} N^{p}(1-s)^{p}\\
&& \ll\left( N^4 (1-s)^4 + N^2 (1-s)^3 \right) \sum_{p + l_1 + ... + l_N = n}\binom{n}{p, l_1, ..., l_N} N^{p}(1-s)^{p} \nonumber \\
&& \ll \left( N^4 (1-s)^4 + N^2 (1-s)^3 \right) \sum_{p + l_1 + ... + l_N = n}\binom{n}{p, l_1, ..., l_N} [N(1-s)]^{p}1^{l_1}1^{l_2}...1^{l_N}\\
&& \ll \left( N^4 (1-s)^4 + N^2 (1-s)^3 \right) (N(1-s) + 1 + ... + 1)^{n} \nonumber \\
&& \ll \left( N^4 (1-s)^4 + N^2 (1-s)^3 \right) (N(1-s) + N)^{n}\\
&& \ll \left( N^4 (1-s)^4 + N^2 (1-s)^3 \right) N^n \nonumber \\
&& \ll N^{n+4}(1-s)^4 + N^{n + 2} (1-s)^{3}.
\end{eqnarray}

We use this inequality in (\ref{h2}) and we obtain that

\begin{equation}\label{h3}
h(n, s, N) = \frac{(s-1)^n}{s^n}N^n + \mathcal{O}(N^{n+2} (1-s)^{n+1})+\mathcal{O}\left( N^{n+4}(1-s)^4\right) + \mathcal{O}\left(N^{n + 2} (1-s)^{3} \right),
\end{equation}
which yields the following theorem.

\begin{theorem}
If $n = 2$, we have that

\begin{equation}
    h(n = 2, s, N) = N^2(1-s)^2 + \mathcal{O}\left( N^{6}(1-s)^4\right) +  \mathcal{O}(N^{4} (1-s)^3),
\end{equation}
and if $n\geq 3$, we have that 

\begin{equation}
    h(n, s, N) = \mathcal{O}\left( N^{n+4}(1-s)^4\right) + \mathcal{O}(N^{n+2} (1-s)^3).
\end{equation}

\end{theorem}
\begin{proof}
We just have to use (\ref{h3}) and 
\begin{equation}
    \frac{(s-1)^n}{s^n} = (1-s)^n + \mathcal{O}((1-s)^{n+1}).
\end{equation}
\end{proof}

\subsection{Case $n= 1$}

We have that
\begin{eqnarray}
h(1, s, N) &&= \int_{[0, \pi]^N}(s-1)\sum_{j=1}^N \frac{2s-2\cos(\theta_j)}{1+s^2 - 2s\cos(\theta_j)} d\mu(\theta)\\
&& = (s-1)\int_{0}^{\pi} \frac{2s - \cos(\theta)}{1 + s^2 - 2s\cos(\theta)} R_1(\theta)d\theta.
\end{eqnarray}

We use
\begin{equation}
\frac{2s-2\cos(\theta)}{1+s^2 - 2s\cos(\theta)} = \frac{1}{s}\left(1+  \frac{s^2-1}{(1-s)^2 + 4s\sin^2(\theta/2)}\right)
\end{equation}
and 
\begin{equation}
    \int_{0}^{\pi} R_1(\theta)d\theta = N
\end{equation}
to state that

\begin{equation}
    h(1, s, N) = \frac{N(s-1)}{s} + \frac{(s+1)(1 - s)^2}{s} \int_{0}^{\pi} \frac{1}{(1-s)^2 + 4s\sin^2(\theta/2)} R_1(\theta)d\theta.
\end{equation}

We can rewrite the integral as 

\begin{equation}
\int_{0}^{\pi} \frac{1}{4s\sin^2(\theta/2)} R_1(\theta)d\theta + \\\int_{0}^{\pi} \left(\frac{1}{(1-s)^2 + 4s\sin^2(\theta/2)}-\frac{1}{4s\sin^2(\theta/2)}\right)R_1(\theta)d\theta.
\end{equation}

So, 

\begin{multline}
    h(1, s, N) = \frac{N(s-1)}{s} + \frac{(s+1)}{4s^2}(1-s)^2 \int_{0}^{\pi} \frac{R_1(\theta)}{\sin^2(\theta/2)}d\theta\\
    -\frac{(s+1)(1-s)^4}{s}\int_{0}^{\pi} \frac{1}{(1-s)^2 + 4s\sin^2(\theta/2)}\frac{R_1(\theta)}{\sin^2(\theta/2)}d\theta.
\end{multline}

We know that $\frac{R_1(\theta)}{\sin^2(\theta/2)}\ll N^3$ from Lemma \ref{boundingR_1} and $\int_0^{\pi} \frac{d\theta}{(1-s)^2 + 4s\sin^2(\theta/2)}\ll \frac{1}{1-s}$ from Lemma \ref{integralI} and so, 

\begin{equation}\label{h4}
    h(1, s, N) = \frac{N(s-1)}{s} + \frac{(s+1)}{4s^2}(1-s)^2 \int_{0}^{\pi} \frac{R_1(\theta)}{\sin^2(\theta/2)}d\theta + \mathcal{O}(N^3 (1-s)^3).
\end{equation}

The following lemma enables us to compute the above integral.

\begin{lemma}
\begin{equation}
\int_{0}^{\pi} \frac{R_1(\theta)}{\sin^2(\theta/2)}d\theta = 2N^2.
\end{equation}
\end{lemma}

\begin{proof}
We have seen that

\begin{equation}
    R_1(\theta) = \frac{1}{2\pi} \sum_{k=0}^{2N-1} \left(1-e^{i\theta(2k+1 - 2N)}\right).
\end{equation}

The result must be real, so we rewrite

\begin{eqnarray}
R_1(\theta) && =\frac{1}{2\pi} \sum_{k=0}^{2N-1}(1 - \cos((2k+1 - 2N)\theta))\\
&& = \frac{1}{\pi} \sum_{k=0}^{2N-1} \sin^2((2k+1 - 2N)\theta/2).
\end{eqnarray}

We use this to state that

\begin{equation}\label{integral}
\int_0^{\pi} \frac{R_1(\theta)}{\sin^2(\theta/2)}d\theta = \frac{1}{\pi}\sum_{k=0}^{2N-1} \int_0^{\pi} \frac{\sin^2((2k+1-2N)\theta/2)}{\sin^2(\theta/2)}d\theta.
\end{equation}

Moreover, if $A\in \N^*$, we know that

\begin{equation}
    \sum_{k=0}^{A-1}e^{ik\theta} = \frac{\sin(A \theta/2)}{\sin(\theta/2)}e^{i \theta/2(A-1)}
\end{equation}

and so that

\begin{equation}
    \frac{\sin(A \theta/2)}{\sin(\theta/2)} = \sum_{k=0}^{A-1} e^{i\theta(k - A/2 + 1/2)}.
\end{equation}

Since the sine function is an odd function, we have that

\begin{equation}
    \left(\frac{\sin(A \theta/2)}{\sin(\theta/2)}\right)^2 = \left(\sum_{k=0}^{|A|-1} e^{i\theta(k - |A|/2 + 1/2)}\right)^2.
\end{equation}
even if $A$ is a negative integer.

We use the above formula in (\ref{integral}) with $A = 2k + 1 - 2N$ and we have that

\begin{eqnarray}
\int_0^{\pi} \frac{R_1(\theta)}{\sin^2(\theta/2)}d\theta&&=\frac{1}{\pi} \sum_{k=0}^{2N-1} \int_{0}^{\pi} \left(\sum_{j = 0}^{|2k+1 - 2N|-1} e^{i \theta(j - |2k+1 - 2N|/2+1/2)}\right)^2d\theta\\
&& = \frac{1}{\pi} \sum_{k=0}^{2N-1} \sum_{\underset{0\leq j_2\leq |2k+1 - 2N|-1}{0\leq j_1\leq |2k+1 - 2N|-1}} \int_0^{\pi} e^{i\theta(j_1 + j_2 - |2k+1 - 2N| + 1)}d\theta \nonumber \\
&& = \frac{1}{\pi} \sum_{k=0}^{2N-1} \sum_{\underset{0\leq j_2\leq |2k+1 - 2N|-1}{0\leq j_1\leq |2k+1 - 2N|-1}} \int_0^{\pi} \cos(\theta(j_1 + j_2 - |2k+1 - 2N| + 1))d\theta
\end{eqnarray}
because the result must be real.

If $n$ is an integer, one can easily check that $\frac{1}{\pi}\int_{0}^{\pi}\cos(n\theta) = \delta_{n}$ i.e. $1$ if $n = 0$ and $0$ otherwise.

Therefore, 

\begin{eqnarray}
\int_0^{\pi} \frac{R_1(\theta)}{\sin^2(\theta/2)}d\theta&& =\sum_{k=0}^{2N-1} \sum_{\underset{0\leq j_2\leq |2k+1 - 2N|-1}{0\leq j_1\leq |2k+1 - 2N|-1}} \delta_{j_1 + j_2 - |2k+1 - 2N| + 1}\\
&& = \sum_{k=0}^{2N-1} |2k+1 - 2N| \nonumber \\
&& = \sum_{k=0}^{N-1}(2N - 2k - 1) + \sum_{k=N}^{2N-1} (2k + 1 - 2N) \nonumber \\
&& = \sum_{k=0}^{N-1}(2N - 2k - 1) + \sum_{k=0}^{N-1} (2k + 1)\\
&& =  \sum_{k=0}^{N-1} 2N \nonumber \\
&& = 2N^2.
\end{eqnarray}

\end{proof}

Using this lemma in (\ref{h4}), we obtain that

\begin{equation}\label{h_1}
    h(1, s, N) = \frac{N(s-1)}{s} + \frac{s+1}{2s^2}N^2(1-s)^2 +\mathcal{O}(N^3(1-s)^3).
\end{equation}

We can now state the main theorem of this subsection.

\begin{theorem}
We have that

\begin{equation}
h(1, s, N) = N(s-1) + N(N-1)(1-s)^2  + \mathcal{O}(N^3(1-s)^3).
\end{equation}

\end{theorem}

\begin{proof}
We expand \eqref{h_1} when $s$ goes to $1$. Then,

\begin{eqnarray}
 && h(1, s, N) =  \frac{N(s-1)}{s} + \frac{s+1}{2s^2}N^2(1-s)^2  + \mathcal{O}(N^3(1-s)^3)\\
&& = \frac{N(s-1)}{1 + (s-1)} + N^2(1-s)^2 + \left(\frac{s+1}{2s^2} - 1\right)N^2(1-s)^2  +  \mathcal{O}(N^3(1-s)^3) \nonumber \\
&& = N(s-1)(1 + (1-s)) + N^2(1-s)^2 + \mathcal{O}(N^3(1-s)^3) \nonumber \\
&&= N(s-1) + N(N-1)(1-s)^2 + \mathcal{O}(N^3(1-s)^3).
\end{eqnarray}
\end{proof}

\subsection{Case $n=0$}

This is the simplest case. We just have 

\begin{equation}
    h(0, s, N) = 1.
\end{equation}

\subsection{Conclusion}

In the last section, we have proved the following theorem which is more precise than Proposition \ref{propintegercase} if $a$ goes to $0$ fast enough.

\begin{theorem}
\begin{equation}
h(n, s, N) = \left\{
    \begin{array}{ll}
        1 & \mbox{if } n =0 \\
        N(s-1) + N(N-1)(1-s)^2  + \mathcal{O}(N^3(1-s)^3) & \mbox{if } n =1 \\
        N^2(1-s)^2 +  \mathcal{O}\left( N^{6}(1-s)^4\right) + \mathcal{O}(N^{4} (1-s)^3) & \mbox{if } n =2 \\
         \mathcal{O}\left( N^{n+4}(1-s)^4\right) + \mathcal{O}(N^{n+2} (1-s)^3) & \mbox{otherwise.}
    \end{array}
\right.
\end{equation}
\end{theorem}

By assuming that $\frac{-1}{2} \leq \frac{2(s-1)}{s+1}N$ and using Proposition \ref{propbound}, we know that

\begin{equation}
(s-1)^K\int_{SO(2N+1)} \left(\frac{\Lambda_X'}{\Lambda_X}(s)\right)^K \dX =  \sum_{n = 0}^{m} \binom{K}{n}h(n, s, N) + \mathcal{O}(h(m+1, s, N))
\end{equation}
where $m$ is odd and $m\geq \lfloor K \rfloor$. 

Let's assume that $m\geq 3$. By applying the above theorem we have arrived at the following theorem.

\begin{theorem}\label{bigtheorem2}
Let $\Lambda_X(s)$ denote the characteristic polynomial of a matrix $X \in SO(2N+1)$, the group of odd orthogonal random matrices of odd dimension and determinant 1, equipped with the Haar measure $\dX$. Let $K\in \mathbb{R}_{+}$ the moment parameter. Assume that $\frac{2(1-s)}{s+1}N \leq \frac{1}{2}$ and take $m$ an odd integer such as $m\geq \lfloor K \rfloor$ and $m\geq 3$. We have that

\begin{multline}\label{bigtheorem}
(s-1)^K\int_{SO(2N+1)} \left(\frac{\Lambda_X'}{\Lambda_X}(s)\right)^K \dX =\\  1 - KN(1-s) - KN(1-s)^2 + \frac{K(K+1)}{2}N^2(1-s)^2  +\mathcal{O}\left( N^{m+5}(1-s)^4\right) +\mathcal{O}(N^{m+3}(1-s)^3).
\end{multline}

Equivalently, we may say

\begin{multline}
\left|(s-1)^K\int_{SO(2N+1)} \left(\frac{\Lambda_X'}{\Lambda_X}(s)\right)^K \dX -\right.\\  \left.\left[1 - KN(1-s) - KN(1-s)^2 + \frac{K(K+1)}{2}N^2(1-s)^2\right]\right|\leq M_K \left( N^{m+5}(1-s)^4 + N^{m+3}(1-s)^3 \right),
\end{multline}

where $M_K>0$ is a constant which only depends on $K$.

\end{theorem}

We can now apply this theorem to the specific case where $s = e^{-a/N}$ where $a$ is a positive function which goes to $0$ as $N$ gets large.

\begin{customthm}{3}
Let $\Lambda_X(s)$ denote the characteristic polynomial of a matrix $X \in SO(2N+1)$, the group of odd orthogonal random matrices of odd dimension and determinant 1, equipped with the Haar measure $\dX$. Let $\alpha = a/N>0$ with $a = o(1)$ as $N \to \infty$, the moment parameter $K \in \mathbb{R}_{+}$, and $m$ be an odd integer such as $m\geq 3$ and $m\geq \lfloor K\rfloor$. Then, 

\begin{multline}
\left(\frac{-a}{N}\right)^K\int_{SO(2N+1)} \left(\frac{\Lambda_X'}{\Lambda_X}(e^{-\alpha})\right)^K \dX \underset{N\to +\infty}{=} \\1 - Ka + \frac{Ka}{2N} + \frac{K(K+1)}{2}a^2 - \frac{K(K+1)a^2}{2N} + \frac{a^2}{N^2}K\frac{3K-1}{24} + \mathcal{O}(N^{m+1} a^4) + \mathcal{O}(N^m a^3).
\end{multline}   
\end{customthm}

To obtain an interesting result, we can ask that $N^m a^3 = o(\frac{a^2}{N^2})$ and $N^{m+1} a^4 = o(\frac{a^2}{N^2})$ i.e. $a = o\left(\frac{1}{N^{m+2}}\right)$ as $N\to \infty$ since $m \geq 3$. And so, we can see that this theorem provides a better result than Theorem \ref{absthm}.

\begin{proof}

First we note the following approximations as $N \to \infty$:

\begin{eqnarray}
1-s &&  = 1 - e^{-a/N} \nonumber \\
&& \underset{N\to +\infty}{=} \frac{a}{N} - \frac{a^2}{2N^2} + \mathcal{O}\left(\frac{a^3}{N^3}\right),
\end{eqnarray} 

\begin{eqnarray}
(1-s)^2 &&= \left(1 - e^{-a/N}\right)^2 \nonumber \\
&& \underset{N\to +\infty}{=} \left(\frac{a}{N} + \mathcal{O}\left(\frac{a^2}{N^2}\right)\right)^2 \nonumber \\
&&  \underset{N\to +\infty}{=} \frac{a^2}{N^2}\left(1 + \mathcal{O}\left(\frac{a}{N}\right)\right)^2 \nonumber \\
&&  \underset{N\to +\infty}{=} \frac{a^2}{N^2} + \mathcal{O}\left(\frac{a^3}{N^3}\right)
\end{eqnarray}

and 

\begin{eqnarray}
\frac{1}{(s-1)^k} && = \frac{1}{(e^{-a/N} - 1)^K} \nonumber \\
&& \underset{N\to +\infty}{=} \frac{1}{(\frac{-a}{N} + \frac{a^2}{2N^2} - \frac{a^3}{6N^3} + \mathcal{O}\left(\frac{a^4}{N^4}\right))^K} \nonumber \\
&& \underset{N\to +\infty}{=} \left(\frac{-N}{a}\right)^K \frac{1}{\left(1 - \frac{a}{2N} + \frac{a^2}{6N^2} + \mathcal{O}\left(\frac{a^3}{N^3}\right)\right)^K} \nonumber \\
&& \underset{N\to +\infty}{=} \left(\frac{-N}{a}\right)^K\left(1 -K\left(\frac{-a}{2N} + \frac{a^2}{6N^2}\right) + \frac{K(K+1)}{2}\left(\frac{a}{2N}\right)^2 + \mathcal{O}\left(\frac{a^3}{N^3}\right)\right) \nonumber \\
&& \underset{N\to +\infty}{=} \left(\frac{-N}{a}\right)^K\left(1 + \frac{Ka}{2N}+ \frac{a^2}{N^2} K\frac{3K - 1}{24} + \mathcal{O}\left(\frac{a^3}{N^3}\right)\right).
\end{eqnarray}
Using these approximations in \eqref{bigtheorem}, we obtain that

\begin{multline}
\left(\frac{-a}{N}\right)^K\int_{SO(2N+1)} \left(\frac{\Lambda_X'}{\Lambda_X}(s)\right)^K \dX \underset{N\to +\infty}{=}\\
\left(1 + \frac{Ka}{2N}+ \frac{a^2}{N^2} K\frac{3K - 1}{24}\right)\left(1 - K N\left(\frac{a}{N} - \frac{a^2}{2N^2}\right) - KN \frac{a^2}{N^2} + \frac{K(K+1)}{2}N^2 \frac{a^2}{N^2}\right) + \mathcal{O}(N^{m+1} a^4) + \mathcal{O}(N^m a^3)
\end{multline}

\begin{eqnarray}
&& \underset{N\to +\infty}{=} \left(1 + \frac{Ka}{2N}+ \frac{a^2}{N^2} K\frac{3K - 1}{24}\right)\left(1 - Ka  + \frac{K(K+1)}{2}a^2- K \frac{a^2}{2N}\right) + \mathcal{O}(N^{m+1} a^4) + \mathcal{O}(N^m a^3) \nonumber \\
&& \underset{N\to +\infty}{=} 1 - Ka + \frac{Ka}{2N} + \frac{K(K+1)}{2}a^2 - \frac{Ka^2}{2N}- \frac{K^2 a^2}{2N} + \frac{a^2}{N^2}K\frac{3K-1}{24} + \mathcal{O}(N^{m+1} a^4) + \mathcal{O}(N^m a^3) \nonumber \\
&& \underset{N\to +\infty}{=} 1 - Ka + \frac{Ka}{2N} + \frac{K(K+1)}{2}a^2 - \frac{K(K+1)a^2}{2N} + \frac{a^2}{N^2}K\frac{3K-1}{24} + \mathcal{O}(N^{m+1} a^4) + \mathcal{O}(N^m a^3).
\end{eqnarray}

In contrast with Theorem \ref{bigtheorem2}, here we must have that $N$ goes to $+\infty$ because we want $\frac{2(1-s)}{s+1}N\leq \frac{1}{2}$ to hold, which is true for $\frac{2(1-s)}{s+1}N\sim a  = o(1)$.

\end{proof}

\newcommand{\etalchar}[1]{$^{#1}$}

\end{document}